\documentclass[a4paper,UKenglish,cleveref,autoref,thm-restate]{lipics-v2021}

\pdfoutput=1 
\hideLIPIcs  

\usepackage{amsmath,amsfonts}
\usepackage{tikz}
\usetikzlibrary{automata, positioning}

\newcommand{\MDP}{\mathcal{M}}
\newcommand{\States}{S}
\newcommand{\Actions}{\mathit{Act}}
\newcommand{\state}{s}

\newcommand{\infinitepath}{\rho}
\newcommand{\finitepath}{\varrho}

\newcommand{\support}{\mathrm{supp}}

\NewDocumentCommand{\IPaths}{d<>}{\IfValueTF{#1}{\mathsf{IPaths}_{#1}}{\mathsf{IPaths}}}
\NewDocumentCommand{\FPaths}{d<>}{\IfValueTF{#1}{\mathsf{FPaths}_{#1}}{\mathsf{FPaths}}}

\newcommand{\AP}{\mathsf{AP}}

\newcommand{\NBA}{N}
\newcommand{\NBAStates}{Q}
\newcommand{\NBALetters}{\Sigma}
\newcommand{\NBAtransition}{\delta}
\newcommand{\NBAinit}{q_0}
\newcommand{\NBAaccept}{F}

\newcommand{\Tran}{\mathcal{T}}
\newcommand{\Locs}{L}
\newcommand{\Vars}{V}
\newcommand{\locinit}{l_{\text{init}}}
\newcommand{\vecinit}{\theta_{\text{init}}}
\newcommand{\transitions}{\mapsto}
\newcommand{\Guard}{G}
\newcommand{\Update}{U}

\newcommand{\Tranx}{\mathcal{T}^\times}

\newcommand{\Init}{\mathsf{Init}}

\newcommand{\Certificate}{\mathcal{C}}

\title{Omega-regular Verification and Control for Distributional Specifications in MDPs} 

\author{S. Akshay}{Dept of CSE, Indian Institute of Technology Bombay, Mumbai, India}{akshayss@cse.iitb.ac.in}{}{}

\author{Ouldouz Neysari}{Singapore Management University, Singapore \and University of Tehran, Iran}{neysari.ouldouz@ut.ac.ir}{}{}

\author{\DJ or\dj e \v{Z}ikeli\'{c}}{Singapore Management University, Singapore}{dzikelic@smu.edu.sg}{}{}

\authorrunning{S. Akshay, O. Neysari, \DJ. \v{Z}ikeli\'{c}}

\Copyright{S. Akshay, Ouldouz Neysari, \DJ or\dj e \v{Z}ikeli\'{c}} 



\funding{This work was supported in part by the Singapore Ministry of Education
(MOE) Academic Research Fund (AcRF) Tier 1 grant (Project ID:22-SIS-SMU-100), Google Research Award 2023 and the SBI Foundation Hub for Data Science \& Analytics, IIT Bombay.}

\nolinenumbers 

\EventEditors{Patricia Bouyer and Jaco van de Pol}
\EventNoEds{2}
\EventLongTitle{36th International Conference on Concurrency Theory (CONCUR 2025)}
\EventShortTitle{CONCUR 2025}
\EventAcronym{CONCUR}
\EventYear{2025}
\EventDate{August 26--29, 2025}
\EventLocation{Aarhus, Denmark}
\EventLogo{}
\SeriesVolume{348}
\ArticleNo{29}

\begin{document}

\maketitle

\begin{abstract}
A classical approach to studying Markov decision processes (MDPs) is to view them as state transformers. However, MDPs can also be viewed as distribution transformers, where an MDP under a strategy generates a sequence of probability distributions over MDP states. This view arises in several applications, even as the probabilistic model checking problem becomes much harder compared to the classical state transformer counterpart. It is known that even distributional reachability and safety problems become computationally intractable (Skolem- and positivity-hard). To address this challenge, recent works focused on sound but possibly incomplete methods for verification and control of MDPs under the distributional view. However, existing automated methods are applicable only to distributional reachability, safety and reach-avoidance specifications.

In this work, we present the first automated method for verification and control of MDPs with respect to distributional omega-regular specifications. To achieve this, we propose a novel notion of distributional certificates, which are sound and complete proof rules for proving that an MDP under a distributionally memoryless strategy satisfies some distributional omega-regular specification. We then use our distributional certificates to design the first fully automated algorithms for verification and control of MDPs with respect to distributional omega-regular specifications. Our algorithms follow a template-based synthesis approach and provide soundness and relative completeness guarantees, while running in PSPACE. Our prototype implementation demonstrates practical applicability of our algorithms to challenging examples collected from the literature.
\end{abstract}

\section{Introduction}\label{sec:intro}

Markov decision processes (MDPs) are a standard model for reasoning and sequential decision making in the presence of uncertainty. The verification community has long studied MDPs as state transformers, where their semantics are interpreted over cylinder sets of paths (see e.g.~\cite{BaierK08}). As a result, quantitative verification questions focus on state-based properties, such as the eventual reachability of a state with maximum probability over all MDP strategies. There is a rich body of literature on efficient algorithms for reasoning about state-based properties in MDPs, including model checking over expressive logics such as PCTL*~\cite{KwiatkowskaNP11}.

An orthogonal class of objectives, which forms our focus in this paper, considers properties that are defined over the space of probability distributions over MDP states, rather than the state space of the MDP. This allows one to reason about the movement of the probability mass, for instance, one can say that always in the future the probability mass is equally divided between two bi-stable states. Such objectives, that we call \emph{distributional objectives}, are simpler to reason about under alternative semantics which view MDPs as {\em distribution transformers}. In this view, starting from some
initial distribution over MDP states, the MDP under a strategy induces a sequence of distributions over MDP states, generating a new distribution at each time step.  One can then specify properties with respect to this sequence of distributions, such as distributional reachability or safety. This view naturally arises in several applications, including multi-agent systems~\cite{BaldoniBMR08,AkshayCMZ24} or biochemical reaction networks~\cite{DBLP:conf/qest/KorthikantiVAK10,HenzingerMW09}. However, it turns out that even the simplest distributional properties such as distributional reachability and safety cannot be expressed in PCTL*~\cite{DBLP:journals/logcom/BeauquierRS06}, rendering classical probabilistic model checking algorithms inapplicable to reasoning about distributional specifications. This means that reasoning about distributional properties in MDPs requires new methods.
%
%
The past decade has seen a rich line of theoretical work on analyzing Markov chains and MDPs as distribution transformers~\cite{DBLP:journals/tse/KwonA11,DBLP:conf/qest/KorthikantiVAK10, DBLP:conf/qest/ChadhaKVAK11,DBLP:conf/lics/AkshayGV18,AkshayCMZ23,AkshayCMZ24,DBLP:journals/jacm/AgrawalAGT15,DBLP:journals/tac/GaoAXJ24}. 
However, existing automated methods are restricted to distributional reachability, safety, or reach-avoidance specifications.

In this paper, we present the first automated method for strategy verification and synthesis in MDPs with respect to {\em distributional $\omega$-regular specifications}, significantly extending the class of distributional objectives for which automated methods are available. In doing so, we focus on the verification problem for a given MDP strategy, as well as the control problem which asks to synthesize an MDP strategy which ensures that a distributional $\omega$-regular specification is satisfied. Our strategy verification and synthesis methods are based on the novel notion of {\em distributional certificates} which we introduce in this work. Distributional certificates provide a sound and complete proof rule for proving that an MDP under a distributionally memoryless strategy satisfies the distributional $\omega$-regular specification of interest. We restrict our attention to {\em distributionally memoryless strategies}, which make moves based on the current distribution rather than the state, and which are known to be sufficient for reasoning about distributional reachability and safety~\cite{AkshayCMZ23,AkshayCMZ24}. Our distributional certificates build on the ideas from program verification and certificates such as ranking function~\cite{ColonS01}, invariants~\cite{ColonSS03} or B\"uchi ranking functions~\cite{ChatterjeeGGKZ24}, and bring these ideas to the setting of reasoning about distributional $\omega$-regular specifications in MDPs.

We then present our automated algorithms for strategy verification and synthesis in MDPs with respect to distributional $\omega$-regular specifications. In the setting of distributional objectives, it is known that safety and reachability are already hard, or more precisely, positivity and Skolem-hard~\cite{AkshayAOW15}. The decidability of both are long-standing open problems in linear dynamical systems~\cite{DBLP:conf/rp/OuaknineW12}. As a result, rather than aiming for sound and complete algorithms that would inherently be computationally expensive/infeasible, we adopt a template-based synthesis approach and instead design algorithms that can more efficiently search for distributional certificates and distributionally memoryless strategies that can be expressed in {\em affine arithmetic}. By fixing symbolic affine templates for the certificate and the strategy and by using existing methods for solving quantified formulas over real arithmetic, one is able to reduce the strategy verification and synthesis problems to satisfiability checking in existential theory of the reals, therefore obtaining sound algorithms that run in PSPACE. Furthermore, our algorithms also provide relative completeness guarantees for computing an affine distributional certificate and a memoryless strategy, whenever they exist.

We implement our approach and consider standard benchmarks and examples from the literature, while focusing on several distributional $\omega$-regular specifications for our evaluation. Our results show the practicality of the approach and the potential for future applications.

\smallskip\noindent{\bf Related work.} In addition to the work already mentioned, we discuss a (non-exhaustive list of) few others. Our distributional certificates draw insights from classical certificates for program verification and template-based synthesis algorithms for their computation. Notable examples include ranking functions for proving termination in programs~\cite{ColonS01,ChatterjeeFG16} and invariants for proving safety in programs~\cite{ColonSS03,DBLP:conf/pldi/Chatterjee0GG20}. Our distributional certificates draw insights from B\"uchi ranking functions of~\cite{ChatterjeeGGKZ24} for proving LTL properties in programs. However, there are several important differences. First, we leverage and lift the idea of B\"uchi ranking to the setting of probability distributions over MDP states. Second, our distributional certificates and algorithms for their computation also need to reason about {\em strategies} in MDPs. This is reflected in the following key difference. Our certificates provide soundness and completeness guarantees for all distributional $\omega$-regular specifications and distributionally memoryless strategies, and proving this requires reasoning about reachability under a strategy (see the proof of Theorem~\ref{thm:certificate}). The certificate of~\cite{ChatterjeeGGKZ24}, on the other hand, is complete only for specifications that can be represented via {\em deterministic} B\"uchi automata, if the specification needs to be satisfied from a set of initial states (see Corollary~2 in~\cite{ChatterjeeGGKZ24}).


In the distributional setting, the probabilistic logics defined in~\cite{DBLP:journals/logcom/BeauquierRS06,DBLP:conf/qest/ChadhaKVAK11,DBLP:journals/jacm/AgrawalAGT15} are all orthogonal to the classical semantics, and the model checking techniques developed are not template-based. To the best of our knowledge, none of these works have been automated. The works~\cite{AkshayCMZ23,AkshayCMZ24} propose certificates for distributional reachability, safety and reach-avoidance and design template-based synthesis algorithms for their computation. Our paper follows this line of work and introduces distributional certificates and template-based algorithms for {\em distributional $\omega$-regular specifications}, hence significantly generalizing the class of distributional properties that we can automatically reason about.

Certificates were also used for reasoning about infinite-state probabilistic models such as probabilistic programs under the state-based view. In particular, supermartingale certificates were proposed for qualitative reachability~\cite{SriramCAV,CFNH16:prob-termination}, quantitative reachability, safety and reach-avoidance~\cite{PrajnaJP07,CNZ17,takisaka2021ranking,ChatterjeeGMZ22,ZikelicLHC23}, and most recently for qualitative $\omega$-regular specifications~\cite{DBLP:conf/cav/AbateGR24}. However, these certificates are not, a priori, applicable to the distributional setting.
 
\section{Preliminaries}\label{sec:prelims}

In this section, we recall the basics of probabilistic systems and Markov decision processes. A {\em probability distribution} over a countable set $X$ is a map $\mu: X \rightarrow [0,1]$ such that $\sum_{x \in X} \mu(x) = 1$. The {\em support} of $X$ is defined via $\support(\mu) = \{x \in X \mid \mu(x) > 0\}$. We use $\Delta(X)$ to denote the set of all probability distributions over $X$.


\smallskip\noindent{\bf MDPs.} A {\em Markov decision process (MDP)} is a tuple $\MDP = (\States, \Actions, P)$. We use $\States$ to denote a finite set of {\em states} and $\Actions$ to denote a finite set of {\em actions}. Slightly overloading the notation, for each state $s \in S$, we write $\Actions(s) \subseteq \Actions$ to denote the set of {\em available actions} at $s$. Finally, $P: \States \times \Actions \rightarrow \Delta(\States)$ is a {\em transition function}, assigning to each state $s$ and available action $a \in \Actions(s)$ a probability distribution over the succcessor states. When $|\Actions(s)|=1$ for each state $s$, we say that $\MDP$ is a {\em Markov chain}. 

An {\em infinite path} in an MDP is a sequence $\infinitepath = s_1,a_1,s_2,a_2,\dots \in (\States \times \Actions)^\omega$, such that $a_i \in \Actions(s_i)$ and $P(s_i,a_i)(s_{i+1})>0$ for all $i \in \mathbb{N}$. A {\em finite path} $\finitepath$ in an MDP is a finite prefix of an infinite path that ends in a state. We use $\infinitepath_i$ and $\finitepath_i$ to denote the $i$-th state along an (in)finite path. We use $\IPaths<\MDP>$ and $\FPaths<\MDP>$ to denote the sets of all infinite and finite paths in the MDP $\MDP$, respectively.

\smallskip\noindent{\bf Semantics of MDPs.} The semantics of MDPs are formalized in terms of strategies. A {\em strategy} (or {\em policy}) in an MDP $\MDP$ is a function $\pi: \FPaths<\MDP> \rightarrow \Delta(\Actions)$ which to each finite path (called a {\em history}) assigns a probability distribution over the action to be taken next. We require that, if a finite path $\finitepath \in \FPaths<\MDP>$ ends in a state $s$, then $\support(\pi(\finitepath)) \subseteq \Actions(s)$. A strategy is said to be {\em memoryless} if the probability distribution over actions depends only on the last state of the finite path and not on the whole history, i.e.~if $\pi(\finitepath) = \pi(\finitepath')$ whenever $\finitepath$ and $\finitepath'$ end in the same state. 
For every initial state distribution $\mu_0 \in \Delta(\States)$, an MDP $\MDP$ and a strategy $\pi$ together give rise to a probability space over the set of all infinite paths in the MDP~\cite{BaierK08}. We denote by $\mathbb{P}^\pi_{\mu_0}$ the probability measure and by $\mathbb{E}^\pi_{\mu_0}$ the expectation operators in this probability space, omitting the MDP $\MDP$ from the notation when clear from the context.


\smallskip\noindent{\bf MDPs as distribution transformers.} MDPs are typically regarded as random generators of infinite paths, giving rise to a probability space over the set of all infinite paths in the MDP. Classical probabilistic model checking problems then explore the expected behaviour of these randomly generated infinite paths, giving rise to {\em path properties}~\cite{BaierK08}. However, one can also view MDPs as {\em (deterministic) transformers of distributions}.

Consider an MDP $\MDP$, a strategy $\pi$, and an initial state distribution $\mu_0 \in \Delta(\States)$. For each $i \in \mathbb{N}$ and state $s$, define $\mu_i(s) = \mathbb{P}^\pi_{\mu_0}[\rho \in \IPaths<\MDP> \mid \rho_i = s]$, i.e.~the probability that the $i$-th state of a randomly generated infinite path is $s$. We write $\mu_i = \MDP^\pi(\mu_0,i)$ for the induced probability distribution of the $i$-th state of a randomly generated infinite path. Hence, the MDP $\MDP$, a strategy $\pi$, and an initial state distribution $\mu_0 \in \Delta(\States)$ together give rise to a sequence of probability distributions over the MDP states
\[ \mu_0, \quad \mu_1 = \MDP^\pi(\mu_0,1), \quad \mu_2 = \MDP^\pi(\mu_0,2), \quad \mu_3 = \MDP^\pi(\mu_0,3), \quad \dots \]
One can then study properties of this sequence of distributions. Some examples are {\em distributional reachability} and {\em distributional safety}, which ask if the induced sequence of distributions contains or does not contain an element of some specified set of distributions~\cite{AkshayCMZ23,AkshayCMZ24}. 


\smallskip\noindent{\bf $\omega$-regular specifications.} In this work we will consider $\omega$-regular specifications, which subsume a broad class of specifications such as those belonging to linear temporal logic (LTL) or computation tree logic (CTL)~\cite{BaierK08}. An {\em $\omega$-regular specification} over a finite set $\AP$ of atomic propositions is defined by a {\em non-deterministic B\"uchi automaton (NBA)} $\NBA = (\NBAStates, \NBALetters, \NBAtransition, \NBAinit, \NBAaccept)$, where $\NBAStates$ is a finite set of states, $\NBALetters = 2^{\AP}$ is a finite set of letters, $\NBAtransition: \NBAStates \times \NBALetters \rightarrow 2^{\NBAStates}$ is a (non-deterministic) transition function, $\NBAinit \in \NBAStates$ is the initial state and $\NBAaccept \subseteq \NBAStates$ are accepting states. An infinite word of letters $\sigma_1,\sigma_2,\dots \in \NBALetters^\omega$ is said to be {\em accepting}, if it gives rise to at least one accepting run in $\NBA$, i.e.~if there exists a run $q_0,q_1,q_2,\dots$ such that $q_{i+1} \in \NBAtransition(q_i,\sigma_i)$ for each $i$ and such that $q_i \in F$ for infinitely many $i$. Note that given an LTL formula $\varphi$, it can be converted to an equivalent NBA $\NBA^\varphi$ in exponential time (see e.g., \cite{BaierK08}). In what follows, we will often write examples and benchmarks in LTL as it will be easier and often more intuitive. But for our analysis, we will convert them to their equivalent NBA and reason only about these NBA as the $\omega$-regular specification. 

\smallskip\noindent{\bf Transition systems.} In order to reason about the synchronous evolution of a sequence of distributions over the MDP states and a run in the NBA, we will later introduce the notion of product distributional transition systems in Section~\ref{sec:certificate}. Hence, we here recall the notion of transition systems, which are commonly used to model imperative numerical programs.

An {\em (infinite-state) transition system} is a tuple $\Tran = (\Locs,\Vars,\locinit,\vecinit,\transitions)$, where 
\begin{itemize}
    \item $\Locs$ is a finite set of locations, 
    \item $\Vars$ is a finite set of real-valued variables, 
    \item $\locinit \in \Locs$ is the initial location, \item $\vecinit\subseteq\mathbb{R}^{|\Vars|}$ is the set of initial variable valuations, and 
    \item $\transitions$ is a finite set of transitions of the form $\tau = (l_\tau,l'_\tau,\Guard_\tau,\Update_\tau)$ with $l_\tau$ a source location, $l'_\tau$ a target location, $\Guard_\tau$ a guard which is a boolean predicate over the variables in $\Vars$, and $\Update_\tau: \mathbb{R}^n \rightarrow \mathbb{R}^n$ an update function.
\end{itemize}
A {\em state} in the transition system is a tuple $(l,x) \in \Locs \times \mathbb{R}^{|\Vars|}$ consisting of a location in $\Locs$ and a valuation of variables in $\Vars$. A transition $\tau = (l_\tau,l'_\tau,\Guard_\tau,\Update_\tau)$ is said to be {\em enabled} at a state $(l,x)$ if $l = l_\tau$ and $x \models \Guard_\tau$. An {\em infinite path} (or a {\em run)} in the transition system is a sequence of states $(l_0,x_0),(l_1,x_1),\dots$ with $l_0 = \locinit$, $x_0 \in \Init$, and where for each $i\in\mathbb{N}_0$ there exists a transition $\tau_i=(l_i,l_{i+1},\Guard_\tau,\Update_\tau)$ enabled at $(l_i,x_i)$ such that $x_{i+1} = \Update_\tau(x_i)$. A state $(l,x)$ is said to be {\em reachable}, if there exists an infinite path that contains $(l,x)$. 

\section{Problem Statement}\label{sec:problem}

We now formally define the problems that we consider in this work. Our goal is to design fully automated algorithms for formal verification and control in MDPs with respect to distributional $\omega$-regular specifications. Hence, we first need to formalize the notion of distributional $\omega$-regular specifications. In what follows, let $\MDP = (\States, \Actions, P)$ be an MDP.

\smallskip\noindent{\bf Distributional $\omega$-regular specifications.} Similarly to the classical $\omega$-regular specification setting, we first need to specify a finite set of {\em atomic propositions} $\AP$. We are interested in reasoning about a sequence of distributions induced by an MDP under a strategy. Hence, we let the set $\AP$ consist of finitely many logical formulas of the form $\mathsf{exp}(\mu(s_1), \dots, \mu(s_{|S|})) \geq 0$. Here, $\mathsf{exp}:\mathbb{R}^{|S|} \rightarrow \mathbb{R}$ is an arithmetic expression over the probabilities $\mu(s_1), \dots,\mu(s_{|S|})$ of being in each state of the MDP, where $s_1,\dots,s_{|S|}$ is an arbitrary (but throughout fixed) enumeration of MDP states. In practice, we let $\AP$ contain exactly those atomic propositions that appear in the property that we want to reason about. 
A {\em distributional $\omega$-regular specification} $\varphi$ is then defined by an NBA $\NBA^\varphi = (\NBAStates, \NBALetters, \NBAtransition, \NBAinit, \NBAaccept)$ with $\Sigma = 2^\AP$.


We now define the semantics of distributional $\omega$-regular specifications. Consider a finite set of atomic propositions $\AP$, a distributional $\omega$-regular specification $\varphi$, a strategy $\pi$ and an initial distribution $\mu_0 \in \Delta(S)$ in the MDP $\MDP$. The MDP $\MDP$ under strategy $\pi$ from the initial distribution $\mu_0$ induces an infinite word $\sigma_0,\sigma_1,\sigma_2,\dots$ in the language $2^\AP$ as follows. As defined in Section~\ref{sec:prelims}, an MDP $\MDP$ under strategy $\pi$ from the initial distribution $\mu_0$ induces a sequence $\mu_0,\mu_1,\mu_2,\dots$ of distributions over MDP states. Then, for each $i \in \mathbb{N}_0$, we define the letter $\sigma_i$ as the set of all atomic propositions in $\AP$ that are satisfied at the distribution $\mu_i$, i.e.~$\sigma_i = \{p \in \AP \mid \mu_i \models p\}$, where we use $\models$ to denote proposition satisfaction. We say that the MDP $\MDP$ {\em satisfies} distributional $\omega$-regular specification $\varphi$ under strategy $\pi$ from initial distribution $\mu_0 \in \Delta(S)$, if this infinite word $\sigma_0,\sigma_1,\sigma_2,\dots$ is accepted by the NBA $\NBA^\varphi$.

\smallskip\noindent{\bf Distributionally memoryless strategies.} We restrict our attention to a class of strategies called distributionally memoryless strategies.  A strategy $\pi: \FPaths<\MDP> \rightarrow \Delta(\Actions)$ is said to be {\em distributionally memoryless} if the probability distribution over actions prescribed by the strategy depends only on the current distribution over the MDP states and not on the whole history. Formally, we require that for any initial distribution $\mu_0 \in \Init$ and for any two finite runs $\rho = s_0,a_0,s_1,a_1,\dots,s_n$ and $\rho' = s_0',a_0',s_1',a_1',\dots,s_n'$ that induce the sequences of probability distributions $\mu_0,\mu_1,\dots,\mu_n$ and $\mu_0',\mu_1',\dots,\mu_n'$ with $\mu_n = \mu_n'$, we have $\pi(\rho) = \pi(\rho')$. When the strategy $\pi$ is distributionally memoryless, we write $\MDP^\pi(\mu) = \MDP^\pi(\mu,1)$ to denote an application of a single-step distribution transformer operator.

It was shown in~\cite{AkshayCMZ23,AkshayCMZ24} that distributionally memoryless strategies are sufficient for reasoning about distributional reachability, safety and reach-avoid specifications. That is, for each of these distributional specifications, there exists a strategy in the MDP under which the specification is satisfied if and only if there exists a distributionally memoryless strategy in the MDP under which the specification is satisfied. While this result need not necessarily hold for distributional $\omega$-regular specifications, the restriction will be needed for enabling automated verification and synthesis as they can be represented in a more compact form.



\smallskip\noindent{\bf Problem statement.}  We are now ready to define our strategy verification and synthesis problems. Consider an MDP $\MDP$, a set of initial distributions $\Init \subseteq \Delta(S)$, and a distributional $\omega$-regular specification $\varphi$:
\begin{enumerate}
	\item {\bf Strategy verification problem.} Given a distributionally memoryless strategy $\pi$, verify that the MDP $\MDP$ satisfies $\varphi$ under $\pi$ from every initial distribution $\mu_0 \in \Init$.
	\item {\bf Strategy synthesis problem.} Compute a distributionally memoryless strategy $\pi$, such that the MDP $\MDP$ satisfies $\varphi$ under $\pi$ from every initial distribution $\mu_0 \in \Init$.
\end{enumerate}

\begin{figure}[t]
\centering
\begin{tikzpicture}[shorten >=1pt, node distance=2cm, on grid, auto]
    \node[state] (s1)   {A}; 
    \node[state] (s2) [right=of s1] {B}; 
    \node[state] (s3) [right=of s2] {C}; 

    \path[->]
        (s1) edge [loop above] node {a} (s1)
             edge node {b} (s2)
        (s2) edge node {1} (s3)
        (s3) edge [bend left, below] node {0.5} (s1)
             edge [loop above, above] node {0.5} (s3);
\end{tikzpicture} 
\caption{An MDP which will serve as our running example. The MDP contains three states $\States = \{A,B,C\}$, two actions $\Actions = \{a,b\}$ with $\Actions(A) = \{a,b\}$, $\Actions(B) = \{a\}$, $\Actions(C) = \{a\}$, and its transition function is defined via $P(A,a)(A) = 1$, $P(A,b)(B) = 1$, $P(B,a)(C) = 1$, $P(C,a)(C) = P(C,a)(A) = 0.5$. We consider a singleton initial distribution set $\Init = \{(A: \tfrac{1}{3},\; B: \tfrac{1}{3},\; C: \tfrac{1}{3})\}$.}
\label{fig:running}
\end{figure}
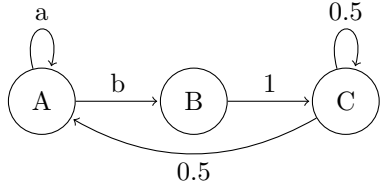

\begin{example}[Running example] \label{ex:running}
The MDP shown in Fig.~\ref{fig:running} was considered in~\cite{AkshayCMZ23} for studying distributional safety specifications and it will serve as our running example. We consider the strategy synthesis and verification problems with respect to the distributional $\omega$-regular specification $\varphi = \mathbf{G}\, \mathbf{F}\, (p(B) \geq 0.249$). For readability, we specify $\varphi$ as an LTL formula over the set of atomic propositions $\AP = \{(\mu(B) \geq 0.249)\}$. This is an example of a {\em distributional persistence specification}, which specifies that the sequence of distributions $\mu_0,\mu_1,\dots$ should contain infinitely many distributions $\mu_i$ with $\mu_i(B) \geq 0.249$.
%
The goal of strategy synthesis is to compute a strategy under which $\varphi$ is satisfied. An example of such a strategy is a (distributionally) memoryless strategy $\pi$ which in state $A$ takes action $b$ with probability $1$. The goal of strategy verification is to verify this claim.  
\end{example}

\begin{remark}[Problem hardness]\label{rmk:hardness}
    The problem of determining if an MDP $\MDP$ under a distributionally memoryless strategy $\pi$ satisfies a distributional $\omega$-regular specification $\varphi$ is computationally hard. It was shown to be Skolem-hard already in the very restricted setting when $\MDP$ is a Markov chain (so the strategy $\pi$ is trivial) and $\varphi$ is a distributional reachability specification for an affine set of goal distributions $\{\mu\in\Delta(S) \mid \mu(s_1) = 0.25\}$~\cite{AkshayAOW15}.
\end{remark}

\begin{remark}[Universal and existential satisfaction problems]\label{rmk:univexist}
	In the terminology of~\cite{AkshayCMZ24} which considered distributional reachability and safety specifications, our problem corresponds to the {\em universal satisfaction} setting, where the specification needs to be satisfied from {\em every} initial distribution $\mu_0 \in \Init$. Dually, one can also consider the {\em existential satisfaction} setting, where the specification needs to be satisfied from {\em at least one} initial distribution $\mu_0 \in \Init$. While we will focus on the universal satisfaction setting for ease of presentation, we also show that all our results straightforwardly extend to the existential satisfaction setting as well.
\end{remark}

\begin{remark}[Memoryless vs distributionally memoryless strategies]\label{rmk:memory}
    Note that distributionally memoryless strategies are {\em not necessarily} memoryless (in the ``classical'' state-based sense). This fact was already shown in~\cite{AkshayCMZ23} for distributional safety specifications, where one may require infinite memory as well as randomized strategies in order to satisfy the specification. This is in stark contrast with the state-based view, where deterministic memoryless strategies are sufficient for specifications such as reachability and safety~\cite{BaierK08}. 
\end{remark}

\section{Certificate for Distributional $\omega$-regular Specifications}\label{sec:certificate}

We now present our sound and complete certificate for proving that an MDP under a distributionally memoryless strategy satisfies some distributional $\omega$-regular specification, which is the main theoretical contribution of this work. In Section~\ref{sec:algo}, we will present our algorithms for automated synthesis and verification of strategies in MDPs with respect to distributional $\omega$-regular specifications, where the certificate will play a central role.

In what follows, we fix an MDP $\MDP = (\States, \Actions, P)$, a set of initial distributions $\Init$, a distributionally memoryless strategy $\pi$, and a distributional $\omega$-regular specification $\varphi$ defined over atomic propositions $\AP$ with an NBA $\NBA^\varphi = (\NBAStates, \NBALetters, \NBAtransition, \NBAinit, \NBAaccept)$ where $\NBALetters = 2^\AP$.


\subsection{Product Distributional Transition System}\label{sec:product}

Recall from Section~\ref{sec:prelims} that, for each initial distribution in $\Init$, the MDP $\MDP$ and the strategy $\pi$ induce a sequence of distributions over the MDP states. This sequence gives rise to an infinite word in the language $2^\AP$ and a run in the NBA $N^\varphi$. In what follows, we introduce product distributional transition systems (PDTS), which will allow us to synchronously reason about the distribution sequence and the NBA run.

\begin{definition}[Product distributional transition system]\label{def:product}
    Let $\MDP = (\States, \Actions, P)$ be an MDP, $\Init$ be a set of initial distributions, $\pi$ be a distributionally memoryless strategy, and $\NBA^\varphi = (\NBAStates, 2^\AP, \NBAtransition, \NBAinit, \NBAaccept)$ be an NBA for some distributional $\omega$-regular specification $\varphi$ defined over atomic propositions $\AP$. 
    A product distributional transition system (PDTS) is a transition system 
    $\Tranx = (\Locs^\times,\Vars^\times,\locinit^\times,\vecinit^\times,\transitions^\times)$, where
\begin{itemize}
    \item $\Locs^\times = \NBAStates$ is the set of states of $\NBA^\varphi$;
    \item $\Vars^\times = \{\mu_1,\dots,\mu_{|\States|}\}$ is a finite state of real-valued variables, with each variable $\mu_i$ corresponding to the probability of being in an MDP state $s_i$;
    \item $\locinit^\times = \NBAinit$ is the initial state of $\NBA^\varphi$;
    \item $\vecinit^\times = \Init$ is the set of initial distributions in $\MDP$; and
    \item $\transitions^\times = \{(q,q',\Guard(\sigma),\MDP^\pi) \mid q,q'\in Q, \sigma\in 2^\AP,q'\in\delta(q,\sigma)\}$, where $\Guard(\sigma) = (\land_{p \in \sigma} p) \land(\land_{p \in \AP \backslash \sigma} \neg p)$ is the predicate defined by atomic propositions contained in $\sigma$, and $\MDP^\pi$ is the linear function defined by the single-step distribution transformer operator of $\MDP$ and~$\pi$.
\end{itemize}
\end{definition}

\begin{example}\label{ex:pdts}
Fig.~\ref{fig:pdts} left shows the NBA for the distributional specification $\varphi = \mathbf{G}\, \mathbf{F}\, (p(B) \geq 0.249)$ considered in Example~\ref{ex:running}. Fig.~\ref{fig:pdts} right then shows the PDTS of our running example MDP in Fig.~\ref{fig:running} and the NBA. The PDTS has the same set of locations $\Locs^\times = \{q_0,q_1\}$ as the NBA and the set of variables $\Vars^\times = \{\mu_1,\mu_2,\mu_3\}$ corresponding to the probabilities of being in MDP states $A, B, C$. The initial location is $\locinit = \NBAinit$ and the set of initial distributions is $\Init = \{(A: \tfrac{1}{3},\; B: \tfrac{1}{3},\; C: \tfrac{1}{3})\}$. Finally, the three PDTS transitions are shown in Fig.~\ref{fig:pdts}.
\end{example}

Note that PDTS indeed models a synchronous execution of a sequence of distributions over MDP states and a run in the NBA. Each infinite path $(q_0,\mu_0),(q_1,\mu_1),\dots$ in the PDTS starts from a state $(q_0,\mu_0) \in \{\NBAinit\}\times\Init$. Then, for each state $(q_i,\mu_i)$ along the infinite path, the next state is obtained by applying some enabled transition $(q_i,q_{i+1},\Guard(\sigma_i),\MDP^\pi)$. For the transition to be enabled, we must have $\sigma_i = \{p \in \AP \mid \mu_i \models p\}$ be the unique letter defined by all atomic propositions satisfied in $\mu_i$. The PDTS then moves to a state $(q_{i+1},\mu_{i+1})$, where $\mu_{i+1} = \MDP^\pi(\mu_i)$ is indeed the next distribution in the sequence and $q_{i+1} \in \delta(q_i,\sigma_i)$ is indeed a successor state in the NBA.

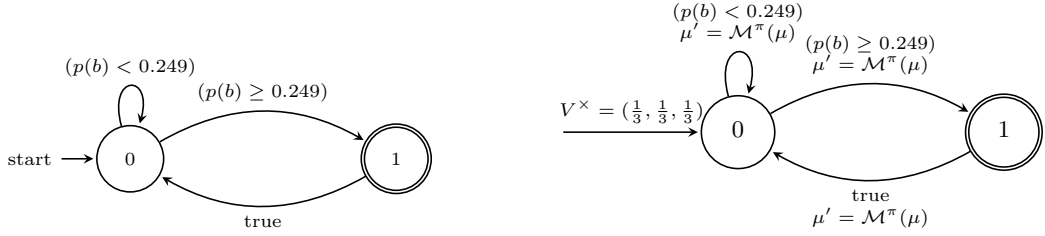
\begin{figure}[t]
\centering

\begin{subfigure}[t]{0.4\textwidth}
    \centering
    
    \tikzset{every node/.style={font=\scriptsize}}\begin{tikzpicture}[->,>=stealth,shorten >=1pt,auto,node distance=3.5cm, semithick]
    
        \node[state,initial] (q0) {$0$};
        \node[state, accepting, right of=q0] (q1) {$1$};

        \path (q0) edge[loop above] node {$(p(b) < 0.249)$} (q0)
              (q0) edge[bend left] node { $(p(b) \geq 0.249)$} (q1)
              (q1) edge[bend left] node {true} (q0);
    \end{tikzpicture}

\end{subfigure}
\hfill
\begin{subfigure}[t]{0.5\textwidth}
    \centering
    \tikzset{every node/.style={font=\scriptsize}}\begin{tikzpicture}[->,>=stealth,shorten >=1pt,auto,node distance=3.5cm, semithick]
    
        \node[state ] (q0) {
          \begin{tabular}{c}
            $0$ \\
        \end{tabular}
        };
        \node[state, accepting, right =2.5cm of q0] (q1) {
            \begin{tabular}{c}
                $1$ 
            \end{tabular}
            };
        \draw 
            (-2.3,0) ->
            node[pos=0.5, above] {$\Vars^\times  = (\frac{1}{3}, \frac{1}{3}, \frac{1}{3})$} 
            (q0);
        \path (q0) edge [loop above,align=center] node { $(p(b) <           0.249)$\\ $\mu' = \MDP^\pi(\mu)$ } (q0)
              (q0) edge[bend left,align=center] node {  $(p(b) \geq     0.249)$ \\  $\mu' = \MDP^\pi(\mu)$ } (q1)
              (q1) edge[bend left,align=center] node { true \\ $\mu' = \MDP^\pi(\mu)$} (q0);              
    \end{tikzpicture}
\end{subfigure}
\caption{The figure on the left shows the NBA for distributional specification $\varphi = \mathrm{GF}(p(b) \geq 0.249)$. The figure on the right then shows the PDTS of the MDP in Fig.~\ref{fig:running} with the  strategy that in state $A$ takes action $b$ with probability $1$, and the NBA on the left. Each PDTS transition is labeled with its guard (top line) and its update function (bottom line). We write $\mu' = \MDP^\pi(\mu)$ as a shorthand notation for $\mu' = \{A: 0.5 \cdot \mu(C), B: \mu(A), C: \mu(B) + 0.5 \cdot \mu(C)\}$.} 
\label{fig:pdts}
\end{figure}


    
An infinite path $(q_0,\mu_0),(q_1,\mu_1),\dots$ in the PDTS is {\em accepting} if $(q_i,\mu_i) \in \NBAaccept \times \Delta(S)$ for infinitely many $i \in \mathbb{N}_0$, i.e.~if it visits states with locations belonging to the accepting set of the NBA infinitely often. The following proposition formalizes the relationship between the satisfaction of a distributional $\omega$-regular specification in the MDP and the existence of an accepting infinite path in the PDTS. We state the proposition for a single initial distribution $\mu_0 \in \Init$, so that it is applicable both in the universal and the existential satisfaction problem settings (see Remark~\ref{rmk:univexist}).

\begin{proposition}\label{prop:product}
    An MDP $\MDP$ with an initial distribution $\mu_0 \in \Delta(S)$ satisfies a distributional $\omega$-regular specification $\varphi$ under a distributionally memoryless strategy $\pi$ if and only if there exists an accepting infinite path in the PDTS $\Tranx = (\Locs^\times,\Vars^\times,\locinit^\times,\vecinit^\times,\transitions^\times)$.
\end{proposition}

\begin{proof}
    Suppose that MDP $\MDP$ with initial distribution $\mu_0$ satisfies distributional $\omega$-regular specification $\varphi$ under distributionally memoryless strategy $\pi$. Let $\mu_0, \mu_1, \dots$ be the sequence of distributions induced by the MDP $\MDP$ under strategy $\pi$, and let $\sigma_0, \sigma_1, \dots$ be the induced infinite word from the initial distribution $\mu_0$. By the definition of satisfiability of distributional $\omega$-regular specifications in Section~\ref{sec:problem}, the infinite word $\sigma_0, \sigma_1, \dots$ is accepted by the NBA $N^\varphi$. Hence, there exists a run $q_0,q_1,\dots$ in $N^\varphi$ such that $q_{i+1} \in \delta(q_i,\sigma_i)$ for each $i \in \mathbb{N}$. But this also means that $(q_0,\mu_0),(q_1,\mu_1),\dots$ is an accepting path in the PDTS $\Tranx = (\Locs^\times,\Vars^\times,\locinit^\times,\vecinit^\times,\transitions^\times)$, which proves one direction of the proposition.

Conversely, suppose that there exists an accepting infinite path $(q_0,\mu_0),(q_1,\mu_1),\dots$ in the PDTS $\Tranx = (\Locs^\times,\Vars^\times,\locinit^\times,\vecinit^\times,\transitions^\times)$. Then, by the definition of transition update functions in the PDTS, we know that $\mu_0,\mu_1,\dots$ is the sequence of distributions induced by the MDP $\MDP$ under strategy $\pi$. Moreover, by the definition of transition guards in the PDTS, we know that $q_{i+1} \in \delta(q_i,\sigma_i)$ for each $i \in \mathbb{N}$ with $\sigma_i$ being the unique letter defined by atomic propositions in $\AP$ satisfied in $\mu_i$. Hence, $q_0,q_1,\dots$ is an infinite run in the NBA $\NBA^\varphi$ induced by the infinite word $\sigma_0,\sigma_1,\dots$. But from the fact that $(q_0,\mu_0),(q_1,\mu_1),\dots$ is an accepting run in the PDTS, it follows that $q_i \in \NBAaccept$ for infinitely many $i \in \mathbb{N}$ and so the infinite word $\sigma_0,\sigma_1,\dots$ is accepted by the NBA $\NBA^\varphi$. By the definition of satisfiability of distributional $\omega$-regular specifications in Section~\ref{sec:problem}, this implies that MDP $\MDP$ with initial distribution $\mu_0$ satisfies specification $\varphi$ under strategy $\pi$, which concludes our proof.
\end{proof}

\subsection{Distributional Certificates}\label{sec:certdefinition}

We are now ready to define our notion of distributional certificates. A {\em distributional certificate} is a pair $(\mathcal{C},I)$ that consists of two components -- a {\em distributional B\"uchi ranking function} $\mathcal{C}$ and a {\em distributional invariant} $I$. The distributional B\"uchi ranking function $\mathcal{C}: \NBAStates \times \Delta(S) \rightarrow \mathbb{R}$ is a function that to each state of the PDTS assigns a real value, which is required to satisfy two conditions. First, the {\em Initial condition} requires the value of $\mathcal{C}$ to be non-negative at all initial states of the PDTS. Second, the {\em B\"uchi ranking condition} requires that, for every reachable state in the PDTS at which the value of $\mathcal{C}$ is non-negative, there exists at least one successor state at which non-negativity is preserved. Furthermore, the value of $\mathcal{C}$ decreases by at least $1$ if the state is not contained in the accepting set of the PDTS. We prove in Theorem~\ref{thm:certificate} below that these two conditions are necessary and sufficient to ensure that, for every initial state $(\NBAinit,\mu_0)$ in the PDTS, there exists an accepting infinite path in the PDTS.

Note that the B\"uchi ranking condition needs to be satisfied only at reachable states of the PDTS. However, the problem of determining the exact set of reachable states is not feasible. Hence, with later automation in mind, we append our certificate with a distributional invariant $I \subseteq \NBAStates \times \Delta(S)$, which is a set that over-approximates the set of reachables states in the PDTS. This is ensured by extending the Initial condition to require that all initial states of the PDTS are contained in the invariant $I$, and by extending the B\"uchi ranking condition to require that the successor state described above is also contained in the invariant $I$. 

The following definition formalizes this intuition. In what follows, for each letter $\sigma \in 2^\AP$ and distribution $\mu \in \Delta(S)$, we write $\mu \models \Guard(\sigma)$ as a shorthand for $\mu \models (\land_{p \in \sigma} p) \land(\land_{p \in \AP \backslash \sigma} \neg p)$.



\begin{definition}[Distributional certificate]\label{def:certificate}
A distributional certificate for an MDP $\MDP$ with a set of initial distributions $\Init$, a distributionally memoryless strategy $\pi$, and a distributional $\omega$-regular specification $\varphi$ with NBA $\NBA^\varphi$, is a tuple $(\Certificate,I)$ consisting of a function $\Certificate: \NBAStates \times \Delta(S) \rightarrow \mathbb{R}$ and a set $I \subseteq \NBAStates \times \Delta(S)$, such that the following conditions hold:
\begin{itemize}
    \item {\em Initial condition.} For all $\mu_0 \in \Init$, we have $\Certificate(q_0,\mu_0) \geq 0$ and $(q_0,\mu_0) \in I$.
    \item {\em B\"uchi ranking condition.} We have the following:
    \begin{itemize}
       \item  {\em Non-negativity at accepting states.} For all NBA states  $q \in \NBAaccept$ and letters $\sigma \in 2^\AP$,
        \begin{equation}\label{eq:cond2}
        \begin{split}
            \forall \mu \in \mathbb{R}^{|S|}.\, \bigvee_{q'\in\delta(q,\sigma)} &\mu \in \Delta(S) \,\land\, \mu \models \Guard(\sigma) \,\land\, \Certificate(q,\mu) \geq 0 \,\land\, (q,\mu) \in I \\
            &\Longrightarrow \Certificate(q',\MDP^\pi(\mu)) \geq 0 \,\land\, (q',\MDP^\pi(\mu)) \in I.
        \end{split}
        \end{equation}
        \item {\em Strict decrease and non-negativity at non-accepting states.} For all NBA states  $q \not\in \NBAaccept$ and letters $\sigma \in 2^\AP$,
        \begin{equation}\label{eq:cond3}
        \begin{split}
            \forall \mu \in \mathbb{R}^{|S|}.\, \bigvee_{q'\in\delta(q,\sigma)} &\mu \in \Delta(S) \,\land\, \mu \models \Guard(\sigma) \,\land\, \Certificate(q,\mu) \geq 0 \,\land\, (q,\mu) \in I \\
            &\Longrightarrow \Certificate(q,\mu) - 1 \geq \Certificate(q',\MDP^\pi(\mu)) \geq 0  \,\land\, (q',\MDP^\pi(\mu)) \in I.
        \end{split}
        \end{equation}
    \end{itemize}
\end{itemize}
\end{definition}

The following theorem establishes that distributional certificates provide a sound and complete proof rule for proving that an MDP under a distributionally memoryless strategy satisfies a distributional $\omega$-regular specification.

\begin{theorem}[Soundness and completeness]\label{thm:certificate}
	An MDP $\MDP$ with a set of initial distributions $\Init$ under a distributionally memoryless strategy $\pi$ satisfies a distributional $\omega$-regular specification $\varphi$ if and only if there exists a distributional certificate for $\MDP$, $\Init$, $\pi$ and $\varphi$.
\end{theorem}

\begin{proof} {\em Soundness.} Suppose that there exists a distributional certificate $(\Certificate, I)$ for $\MDP$, $\Init$, $\pi$ and $\varphi$. To show that $\varphi$ is satisfied, by Proposition~\ref{prop:product} it suffices to show that the PDTS $\Tranx$ admits an accepting infinite path for every initial state in $\{q_0\} \times \Init$. 
	
Fix an initial state $(q_0,\mu_0) \in \{q_0\} \times \Init$. By the Initial condition in Definition~\ref{def:certificate}, we know that $\Certificate(q_0,\mu_0) \geq 0$ and $(q_0,\mu_0) \in I$. Hence, by the B\"uchi ranking condition in Definition~\ref{def:certificate}, we can repeadly select successor states in order to obtain an infinite path $(q_0,\mu_0), (q_1,\mu_1),\dots$ in $\Tranx$ such that, for each $i \in \mathbb{N}_0$, we have
\begin{itemize}
	\item $\Certificate(q_i,\mu_i) \geq 0$ and $(q_i,\mu_i) \in I$, and
	\item whenever $q_i \not\in \NBAaccept$ is not an accepting state in $\NBA^\varphi$, we have $\Certificate(q_i,\mu_i) - 1 \geq \Certificate(q_{i+1},\mu_{i+1})$.
\end{itemize}
We claim that $(q_0,\mu_0), (q_1,\mu_1),\dots$ is an accepting infinite path in $\Tranx$. To prove this, note that for every $(q_i,\mu_i)$ with $q_i \not\in \NBAaccept$, the value of $\mathcal{C}$ needs to keep decreasing by at least $1$ in each subsequent step while also remaining non-negative. Hence, in at most $\lceil \Certificate(q_i,\mu_i) \rceil$ steps, the path must again reach an accepting state. Thus, the infinite path $(q_0,\mu_0), (q_1,\mu_1),\dots$ reaches accepting states in $\NBAaccept \times \Delta(S)$ infinitely many times and is an accepting infinite path. Since the initial state $(q_0,\mu_0) \in \{q_0\} \times \Init$ was arbitrary, this concludes the proof.
		
\smallskip\noindent{\em Completeness.}	Conversely, suppose that $\MDP$ with a set of initial distributions $\Init$ under distributionally memoryless strategy $\pi$ satisfies distributional $\omega$-regular specification $\varphi$. We construct an instance $(\Certificate, I)$ of a distributional certificate for $\MDP$, $\Init$, $\pi$ and $\varphi$ as follows.

Consider an arbitrary but throughout fixed enumeration $q_1,\dots,q_{|Q|}$ of NBA states. We define an operator $\textsc{Next}: \NBAStates \times \Delta(S) \rightarrow \NBAStates \times \Delta(S)$ via
\begin{itemize}
    \item $\textsc{Next}(q,\mu) = (q_i,\MDP^\pi(\mu))$ with $i$ being the smallest index such that $(q,\mu),(q_i,\MDP^\pi(\mu))$ are successor states along some accepting infinite path in the PDTS, if such an accepting infinite path exists, or
    \item $\textsc{Next}(q,\mu) = (q,\mu)$, otherwise.
\end{itemize}
In other words, $\textsc{Next}(q,\mu)$ fixes a successor state of $(q,\mu)$ along some accepting infinite path in the PDTS if such a path exists, or halts the sequence at the state $(q,\mu)$ otherwise. Therefore, the transitive closure of the operator $\textsc{Next}(q,\mu)$ from the set of initial PDTS states $\{q_0\} \times \Init$ allows us to consistently fix a {\em unique accepting infinite path} for each state $(q,\mu)$ that is contained along some accepting infinite path. 

We can now define our distributional certificate $(\Certificate,I)$. Let distributional invariant $I \subseteq \NBAStates \times \Delta(S)$ be the set of all states in the PDTS that are reachable from $\{q_0\} \times \Init$ under the transitive closure of the $\textsc{Next}$ operator. Moreover, for each PDTS state $(q,\mu) \in I$, let $d_{\text{accept}}(q,\mu)$ denote the number of steps when applying the $\textsc{Next}$ operator until an accepting state in $\NBAaccept \times \Delta(S)$ is reached, with $d_{\text{accept}}(q,\mu) = 0$ if $(q,\mu) \in \NBAaccept \times \Delta(S)$ is an accepting state. Finally, let the distributional B\"uchi ranking function $\Certificate$ be defined via
\[ \Certificate(q,\mu) = \begin{cases}
    d_{\text{accept}}(q,\mu), &\text{if } (q,\mu) \in I, \\
    -1, &\text{otherwise}.
\end{cases} \]
We claim that $(\Certificate, I)$ is an instance of a distributional certificate. The Initial condition in Definition~\ref{def:certificate} is satisfied since the MDP $\MDP$ under strategy $\mu$ satisfies $\varphi$, therefore every initial state $(\NBAinit,\mu) \in
\{\NBAinit\} \times \Init$ belongs to some accepting infinite path in the PDTS and so $(\NBAinit,\mu) \in I$ and $d_{\text{accept}}(\NBAinit,\mu) \geq 0$. On the other hand, by the definition of the $\textsc{Next}$ operator and $I$, for each $(q,\mu) \in I$ we have that also $\textsc{Next}(q,\mu) \in I$. Moreover, $\textsc{Next}(q,\mu) = d_{\text{accept}}(q,\mu) - 1 \geq 0$ if $q \not\in \NBAaccept$ and $\textsc{Next}(q,\mu) \geq 0$ id $q \in \NBAaccept$. Hence, the B\"uchi ranking condition in Definition~\ref{def:certificate} is also satisfied, and $(\Certificate,I)$ is a distributional certificate.
\end{proof}

\begin{example}

Consider again the MDP in Fig.~\ref{fig:running}. As in Example~\ref{fig:running}, consider the strategy $\pi$ which in state $A$ takes action $b$ with probability $1$, and the distributional specification $\varphi = \mathbf{G}\, \mathbf{F}\, (p(B) \geq 0.249)$. An NBA for $\varphi$ and the resulting PDTS are shown in Fig.~\ref{fig:pdts}. The following is an example of a distributional certificate $(\mathcal{C},I)$ for $\MDP$, $\Init$, $\pi$ and~$\varphi$:
\begin{equation*}
    \Certificate(q,\mu) = \begin{cases}
        1 + 250 \cdot \mu(A) + 750 \cdot \mu(C), &\text{if } q = q_0, \\
        1.25 - 1.25 \cdot \mu(C), &\text{if } q = q_1,
    \end{cases}
\end{equation*}
and $I = \{(q_0,\mu) \mid 1.25 + \mu(A) + 1 \cdot \mu(B) - \mu(C) \geq 0\} \cup \{(q_1,\mu) \mid \mu(A) + 0.25 \cdot \mu(B)\}$. 

One can verify by inspection that the Initial condition and the B\"uchi ranking condition in Definition~\ref{def:certificate} are satisfied. We note that the above distributional certificate $(\mathcal{C},I)$ is the certificate computed by our prototype implementation in Section~\ref{sec:experiments}.



\end{example}

\begin{remark}[Distributional certificates for the existential problem]
    As discussed in Section~\ref{sec:prelims} and Remark~\ref{rmk:univexist}, our distributional certificate in Definition~\ref{sec:certdefinition} and our soundness and completeness result in Theorem~\ref{thm:certificate} consider the universal satisfaction setting. However, their extension to the existential setting is immediate. The only required change in the definition of distributional certificates is to require the Initial condition in Definition~\ref{sec:certdefinition} to hold {\em for some} initial distribution $\mu_0 \in \Init$. On the other hand, the soundness and completeness proof proceeds analogously as in the proof of Theorem~\ref{thm:certificate}, with the difference in the completeness proof being that the distributional invariant $I$ is defined as the transitive closure of the $\textsc{Next}$ operator with the singleton initial set $\{(\NBAinit,\mu_0)\}$, rather than the initial set $\{\NBAinit\} \times \Init$.
\end{remark}

\section{Template-based Strategy Verification and Synthesis with Certificates}\label{sec:algo}


We now present our algorithms for the strategy verification and synthesis problems for distributional $\omega$-regular specifications. The core of the verification algorithm is to synthesize an {\em affine} distributional certificate in the PDTS of the input MDP and the specification, which proves that the specification is satisfied. When we move from strategy verification to strategy synthesis, we also synthesize an {\em affine} distributionally memoryless strategy. 

The restriction to affine distributional certificates and affine distributionally memoryless strategies is needed to ensure tractability. While Theorem~\ref{thm:certificate} establishes soundness and completeness of our distributional certificates, in combination with Remark~\ref{rmk:hardness} it also implies that giving a sound and complete algorithm for synthesizing distributional certificates is Skolem-hard. Hence, in this section, we instead focus on designing sound and {\em relatively complete} algorithms for synthesizing an affine distributional certificate together with an affine distributionally memoryless strategy (the latter for the synthesis problem), when they exist. 


\smallskip\noindent{\bf Affine distributional certificates and strategies.} We first formalize the notions of affine distributional certificates and distributionally memoryless strategies:
\begin{itemize}
    \item {\bf Affine distributional certificates.} A distributional certificate $(\mathcal{C},I)$ is said to be {\em affine} if both the distributional B\"uchi ranking function $\mathcal{C}$ and the distributional invariant $I$ can be expressed in terms of affine expressions and inequalities over the space $\Delta(S)$ of probability distributions over MDP states. We require $\mathcal{C}$ to be of the form
    \begin{equation}\label{eq:rankingtemplate}
        \mathcal{C}(q,\mu) = \sum_{i=1}^{|S|} a^q_i \cdot \mu(s_i) + b^q,
    \end{equation}
    where $a^q_1,\dots,a^q_{|S|},b^q$ are some real valued coefficients for each NBA state $q \in \NBAStates$. That is, for each NBA state $q$, the function $\mathcal{C}(q,\cdot)$ is an affine function over the probabilities of being in each MDP state, with $\mu(s_1),\dots,\mu(s_{|V|})$ being the variables that capture probabilities of being in each MDP state and $a^q_1,\dots,a^q_{|S|},b^q$ being the coefficients of the affine function. Similarly, we require $I$ to be a set defined by a conjunction of $N_I$ affine inequalities over the probabilities of being in each MDP state, i.e.~to be of the form
    \begin{equation}\label{eq:invarianttemplate}
        I = \Big\{ (\mu,q) \in \Delta(S) \times \NBAStates \mid \land_{k=1}^{N_I} I^k(q,\mu) \geq 0\Big\},
    \end{equation}
    where each $I^k(q,\mu) = \sum_{i=1}^{|S|} c^{k,q}_i \cdot \mu(s_i) + d^{k,q} \geq 0$ and $c^{k,q}_1,\dots,c^{k,q}_{|S|},d^{k,q}$ are some real valued coefficients for each NBA state $q \in \NBAStates$ and each $k \in \{1,\dots,N_I\}$. The number $N_I$ is referred to as the {\em size of the invariant} and will be an algorithm parameter.

    \item {\bf Affine distributionally memoryless strategies.} A distributionally memoryless strategy $\pi: \Delta(S) \rightarrow \Delta(Act)$ is said to be {\em affine}, if for each state $s \in \States$, action $a \in \Actions$ and state distribution $\mu \in \Delta(S)$, the probability of taking action $a$ in state $s$ given the current distribution over states $\mu$ is of the form
    \begin{equation}\label{eq:strategytemplate}
    \pi(s,a)(\mu) = \frac{\sum_{i=1}^{|S|} e^\pi_{i,s,a} \cdot \mu(s_i) + f^\pi_{s,a}}{\sum_{i=1}^{|S|} g^\pi_{i,s} \cdot \mu(s_i) + h^\pi_{s}},
    \end{equation}
    where $e^\pi_{1,s,a},\dots,e^\pi_{|S|,s,a},f^\pi_{s,a}$ and $g^\pi_{1,s},\dots,g^\pi_{|S|,s},h^\pi_{s}$ are real valued constants. The denominator is used in order to normalize the probabilities such that the sum of probabilities of all actions being taken at a state $s$ is $1$.
\end{itemize}

\smallskip\noindent{\bf Algorithm input.} Both our verification and synthesis algorithms take as input an MDP $\MDP = (\States, \Actions, P)$, a set of initial distributions $\Init$, and a distributional $\omega$-regular specification $\varphi$ defined over atomic propositions $\AP$. We assume that the distributional $\omega$-regular specification is provided via an NBA $\NBA^\varphi = (\NBAStates, \NBALetters, \NBAtransition, \NBAinit, \NBAaccept)$ with letters $\NBALetters = 2^\AP$, which accepts the same set of infinite words over $2^\AP$ as $\varphi$. 
Finally, the algorithms also take as input the size of the invariant $N_I$ that needs to be synthesized. The verification algorithm in addition takes as input an affine distributionally memoryless strategy $\pi$.

\smallskip\noindent{\bf Algorithm overview.} Both verification and synthesis algorithms follow a template-based synthesis approach and proceed in four steps. In the first step, the PDTS of the input MDP and the distributional specification is constructed. In the second step, the algorithms fix a symbolic template for the affine distributional certificate, i.e.~symbolic variables for each real valued coefficient in affine expressions that define the certificate. The synthesis algorithm also fixes a symbolic template for the affine distributionally memoryless strategy. In the third step, the algorithms collect a system of constraints over the symbolic template variables, that together encode all defining conditions of distributional certificates in Definition~\ref{def:certificate} as well as conditions for the strategy template to define a valid distributionally memoryless strategy (the latter for the synthesis algorithm). Finally, in the fourth step, the collected system of constraints is solved by using an SMT-solver, to get a concrete valuation of the symbolic template variables which in turn gives rise to a distributional certificate and a distributionally memoryless strategy. 
In what follows, we detail each of these four steps.

\smallskip\noindent{\bf Step 1: Constructing the PDTS.} In this step, the PDTS $\Tranx = (\Locs^\times,\Vars^\times,\locinit^\times,\vecinit^\times,\transitions^\times)$ is constructed from the given MDP $\MDP$ and the NBA $N^\varphi$, as explained in Section~\ref{sec:product}. 

\smallskip\noindent{\bf Step 2: Fixing templates.}
    The algorithms fix a template for the affine distributional certificate $(\mathcal{C},I)$, while the synthesis algorithm also fixes a template for the affine distributionally memoryless strategy $\pi$. The novelty, compared to prior work on verification and synthesis for distributional reachability and safety specifications~\cite{AkshayCMZ23,AkshayCMZ24}, lies in a more complex template design for the distributional certificate, which is now defined with respect to the PDTS:
\begin{itemize}
	\item {\bf Template for $\Certificate$.} Recall that an affine distributional B\"uchi ranking function is of the form $\mathcal{C}(q,\mu) = \sum_{i=1}^{|S|} a^q_i \cdot \mu(s_i) + b^q$ as in eq.~\eqref{eq:rankingtemplate}. Hence, the template for $\mathcal{C}$ is defined by introducing a set of symbolic template variables $a^q_1,\dots,a^q_{|S|},b^q$ for each NBA state $q \in \NBAStates$.
	\item {\bf Template for $I$.} Similarly, the template for an affine distributional invariant $I$ is of the form as in eq.~\eqref{eq:invarianttemplate}. Hence, the template for $I$ is defined by introducing a set of symbolic template variables $c^{k,q}_1,\dots,c^{k,q}_{|S|},d^{k,q}$ for each NBA state $q \in \NBAStates$ and each $k \in \{1,\dots,N_I\}$, where $N_I$ is the algorithm parameter that specifies the size of the invariant.
    \item {\bf Template for $\pi$ (synthesis algorithm).} The template for an affine distributionally memoryless strategy $\pi$ is of the form as in eq.~\eqref{eq:strategytemplate}, hence it is defined by introducing symbolic template variables $e^\pi_{1,s,a},\dots,e^\pi_{|S|,s,a},f^\pi_{s,a}$ and $g^\pi_{1,s},\dots,g^\pi_{|S|,s},h^\pi_{s}$ for each state $s \in \States$ and action $a \in \Actions$.
    Note that, in the special case when we are interested in synthesizing memoryless strategies instead of distributionally memoryless strategies, the strategy template becomes simpler. Instead of the template as in eq.~(\ref{eq:strategytemplate}), we introduce a single symbolic template variable $p_{s,a}^\pi$ for each state-action pair $s \in \States$ and $a \in \Actions$, to encode the probability of taking action $a$ in state $s$.
\end{itemize}

\smallskip\noindent{\bf Step 3: Collecting constraints.} In this step, the algorithms collect a system of constraints over the symbolic template variables that together encode that $\mathcal{C}$ and $\mathcal{I}$ indeed define a valid distributional certificate. For the synthesis algorithm, we also collect a system of constraints that encode that $\pi$ defines a valid distributionally memoryless strategy. In each of the following constraints, each appearance of $\mathcal{C}$, $\mathcal{I}$ and $\pi$ is replaced by the symbolic template introduced in Step~2, in the form as in eq.~\eqref{eq:rankingtemplate}, \eqref{eq:invarianttemplate} and \eqref{eq:strategytemplate}. Moreover, we write $\mu \in \Delta(S)$ for the conjunction of affine inequalities $\land_{i=1}^{|S|} (\mu_i \geq 0) \land (\mu_1+\dots+\mu_{|S|}=1)$.
\begin{itemize}
    \item {\bf Initial condition.}  We define 
    \[ \Phi_{\text{init}} \equiv \forall \mu \in \mathbb{R}^{|S|}.\, \mu \in \Delta(S) \,\land\, \mu \in \Init \Longrightarrow \mathcal{C}(\NBAinit,\mu) \geq 0 \,\land\, \bigwedge_{k=1}^{N_I}I^k(\NBAinit,\mu) \geq 0. \]
    \item {\bf B\"uchi ranking condition for accepting states.} For each accepting state $q \in \NBAaccept$ and letter $\sigma \in 2^\AP$ in the NBA, we define 
    \begin{equation*} 
    \label{eq:buchiacc}
    \begin{split}
        \Phi_{\text{B\"uchi},q,\sigma} \equiv \forall \mu \in \mathbb{R}^{|S|}.\, &\bigvee_{q'\in\delta(q,\sigma)} \mu \in \Delta(S) \,\land\, \mu \models \Guard(\sigma) \,\land\, \Certificate(q,\mu) \geq 0 \,\land\, \bigwedge_{k=1}^{N_I} I^k(q,\mu) \geq 0 \\
        &\Longrightarrow \Certificate(q',\MDP^\pi(\mu)) \geq 0 \,\land\, \bigwedge_{k=1}^{N_I} I^k(\MDP^\pi(q',\mu)) \geq 0.
    \end{split}
    \end{equation*}
    \item {\bf B\"uchi ranking condition for nonaccepting states.} For each non-accepting state $q \in \NBAStates\backslash\NBAaccept$ and letter $\sigma \in 2^\AP$ in the NBA in the NBA, we define 
    
    \begin{equation*}
    \label{buchinonacc}
    \begin{split}
    	\Phi_{\text{B\"uchi},q,\sigma} \equiv \forall \mu \in \mathbb{R}^{|S|}.\, &\bigvee_{q'\in\delta(q,\sigma)} \mu \in \Delta(S) \,\land\, \mu \models \Guard(\sigma) \,\land\, \Certificate(q,\mu) \geq 0 \,\land\, \bigwedge_{k=1}^{N_I} I^k(q,\mu) \geq 0 \\
    	&\Longrightarrow \Certificate(\mu,q) - 1 \geq \Certificate(\MDP^\pi(\mu),q') \geq 0 \,\land\, \bigwedge_{k=1}^{N_I} I^k(\MDP^\pi(q',\mu)) \geq 0.
    \end{split}
    \end{equation*}
        \item {\bf Strategy conditions (synthesis algorithm).} For the strategy template to indeed define a valid affine distributionally memoryless strategy, we require that:
        \begin{equation*}
        	 \label{strategy}
        	\Phi_{\pi} \equiv \bigwedge_{s\in S} \Big( \sum_{a\in Act} \pi(s,a)(\mu)=1 \land \bigwedge_{a\in Act} (\pi(s,a)(\mu) \geq 0) \Big).
        \end{equation*}
\end{itemize}
In the above definitions, note that $\MDP^\pi(\mu)$ is the one-step successor from distribution $\mu$ when policy $\pi$ is applied in the MDP, computed as: $\sum_{\state \in S, a \in Act(s)} \pi(s,a)(\mu) \cdot P(s, a)$.


\smallskip\noindent{\bf Step 4: Constraint solving.} The strategy condition constraint $\Phi_\pi$ is a purely existentially quantified Boolean combination of affine inequalities over the symbolic template variables. However, constraints $\Phi_{\text{init}}$ and $\Phi_{\text{B\"uchi},q,\sigma}$ are all of the form
\[ \forall \mu \in \mathbb{R}^{|S|}.\, \bigvee_{i=1}^m\bigwedge_{j=1}^n \text{aff-expr}_{i,j}(t,\mu) \geq 0 \geq \bigwedge_{l=1}^k \text{poly-expr}_k(t,\mu) \geq 0, \] 
where $t$ is the vector of all symbolic template variables, $\text{aff-expr}_{i,j}(t,\mu)$'s are some affine functions and $\text{poly-expr}_k$'s are some polynomial functions over the vectors of variables $t$ and $\mu$. Polynomial expressions on the right hand side arise due to the quotients of affine expressions that define affine distributionally memoryless strategies, see eq.~\eqref{eq:strategytemplate}. Hence, multiplying both sides of the inequality by the affine expressions appearing in denominators results in polynomial expressions over variables in $t$ and $\mu$.

The problem of synthesizing affine distributional certificates and affine distributionally memoryless strategies then reduces to solving a system of constraints that contain quantifier alternation $\exists t.\, \forall \mu$. Such quantifier alternation over real-valued variables is generally hard to handle directly and can lead to inefficiency in solvers. In order to allow for a more efficient constraint solving, before passing our system of constraints to an SMT-solver, we first apply Handelman's theorem~\cite{handelman1988representing} to translate $\Phi_{\text{init}}$ and $\Phi_{\text{B\"uchi},q,\sigma}$ into a purely existentially quantified system of polynomial constraints over the symbolic template variables in $t$ and auxiliary variables introduced by the translation, whose satisfiability implies satisfiability of the original constraints. This translation is common in template-based program analysis, see~\cite{AsadiC0GM21} for details. This step allows for more efficient constraint solving as well as better bound on the algorithm complexity. Finally, the resulting purely existentially quantified system of polynomial constraints over real-valued variables is solved via an SMT solver.

In the special case when we are interested in synthesizing memoryless strategies rather than distributionally memoryless strategies, we may use Farkas' lemma~\cite{farkas1902theorie} rather than Handelman's theorem. This yields a {\em sound and complete} translation into an equisatisfiable purely existentially satisfied system of constraints.

\smallskip\noindent{\bf Soundness, relative completeness, complexity.} Soundness of our algorithms follows from the soundness of all four steps above, including soundness of the transformations via Handelman's theorem and Farkas' lemma~\cite{AsadiC0GM21}. Since the Farkas' lemma transformation leads to an equisatisfiable system of constraints, it also follows that our algorithm is {\em relatively complete} -- it is guaranteed to synthesize an affine distributional certificate and memoryless strategy whenever they exist. Finally, our algorithms provide a PSPACE complexity upper bound as they reduce the synthesis and verification problems to solving a sentence in the existential first-order theory of the reals. The following theorem summarizes these results.

\begin{theorem}\label{thm:finalalgo}
	\emph{Soundness}: If the algorithm returns an affine distributional certificate $(\Certificate,I)$ and an affine distributionally memoryless strategy $\pi$ (for the synthesis algorithm), then the MDP $\MDP$ with initial distributions $\Init$ under strategy $\pi$ satisfies specification $\varphi$.
    
	
	\emph{Relative completeness}: If there exist an affine distributional certificate $(\Certificate,I)$ and a memoryless strategy $\pi$, then there exists an invariant size $N_I \in \mathbb{N}$ such that $(\Certificate,I)$ and $\pi$ are computed by the algorithm.
	
	\emph{Complexity}: The runtime of the algorithm is in PSPACE in the size of the encoding of the MDP, NBA $N^\varphi$, startegy $\pi$ (for the verification algorithm) and invariant size $N_I \in \mathbb{N}$. 
\end{theorem}

\section{Experimental Evaluation}\label{sec:experiments}

We implemented a prototype of our method in Python~3 and experimentally evaluated it on a number of challenging verification and synthesis tasks collected from the literature on distributional specifications in MDPs. Our prototype takes as input an MDP (in the Prism~\cite{KwiatkowskaNP11} input format) and an LTL specification. The LTL specification is then translated into an NBA via Spot~\cite{duret2022spot}. For the constraint solving step in our algorithms, we use PolyQEnt~\cite{chatterjee2025polyqent} which provides a tooling support for quantifier elimination via Farkas' lemma and Handelman's theorem. PolyQEnt uses Z3~\cite{MouraB08} and MathSAT5~\cite{mathsat5} as backend SMT solvers for the final system of purely existentially quantified constraints. We set the invariant size parameter to $N_I = 1$, which was sufficient for all our experiments. Our experiments were conducted on consumer-grade hardware (AMD Ryzen 5 5625U CPU, 8GB RAM). 


\smallskip\noindent {\bf Benchmarks.} We evaluated our method on several examples collected from the literature:
\begin{itemize}
    \item {\bf GridWorld (synthesis).} Motivated by~\cite{AkshayCMZ24}, these benchmarks model robot swarms in gridworld environments. Initially, all robots are placed in the top-left corner of the gridworld environment. Some of the cells are covered by walls whereas some are slippery and with certain probability may lead to moving in an undesired direction. Hence, each environment induces an MDP. As shown in~\cite{AkshayCMZ24}, the evolution of a robot swarm can be analyzed by taking the distribution transformer view of MDPs and considering how the robots are distributed across the gridworld cells at each time step. In Table~\ref{tab:example}, we consider $5$ gridworld benchmarks of varying sizes and consider two distributional specifications: (1)~at least 90\% of robots should be in some slippery target cell infinitely often, and (2) in addition, at most 50\% of robots should occupy some narrow passage at any point in time.
    \item {\bf PageRank (verification).} We consider a Markov chain representation of the PageRank algorithm taken from~\cite{DBLP:journals/jacm/AgrawalAGT15}. Given the context, we consider various verification tasks, which are of the form: always if the probability mass at some vertex/page is above a threshold, then eventually, it must be above a threhold in another page.
    \item {\bf Pharamakocinetics (verification)} We also consider a 6 state Markov chain from ~\cite{DBLP:journals/jacm/AgrawalAGT15} which is adapted from a Pharmacokinetics example in~\cite{DBLP:conf/qest/ChadhaKVAK11}. We use the two queries that were listed in~\cite{DBLP:journals/jacm/AgrawalAGT15} as the motivating examples to obtain our specifications. 
    \item {\bf Benchmarks from~\cite{AkshayCMZ23} (verification and synthesis)}. Finally, we collect $3$ pairs of verification and synthesis tasks from~\cite{AkshayCMZ23}. In the verification task a strategy is fixed, whereas in the synthesis task one also needs to compute the strategy. While~\cite{AkshayCMZ23} considered distributional safety specifications, we design more complex $\omega$-regular specifications.
\end{itemize}



\noindent{\bf Results.} Our experimental results are shown in Table~\ref{tab:example}. Our results demonstrate that our prototype is able to solve a number of challenging verification and synthesis tasks for distributional $\omega$-regular specifications in MDPs, that were beyond the reach of all existing methods. This is achieved at runtimes that are comparable or even lower than those reported by earlier methods for distributional reachability and safety specifications in~\cite{AkshayCMZ23,AkshayCMZ24} on benchmarks of similar size. Hence, even though we consider a significantly more general class of distributional $\omega$-regular specifications, our algorithms do not lead to a significant computational overhead. Moreover, for all our strategy synthesis tasks, our prototype was able to compute {\em memoryless strategies} that lead to distributional specification satisfaction. This demonstrates the generality of relative completeness guarantees provided by our algorithms.

We also make some observations. First, as can be seen from runtimes reported in Table~\ref{tab:example}, the final SMT-solving step is computationally the most expensive step of our algorithms. Constraint generation took at most a few seconds in all cases. Second, we observe that strategy synthesis tasks are generally computationally more expensive, which is expected given that they require synthesizing both the strategy and the distributional certificate. However, in some cases the synthesis problems can also be solved more efficiently. This is demonstrated by the last $6$ experiments (CAV'23 in Table~\ref{tab:example}), where synthesis is achieved at lower runtimes due to our prototype being able to compute a simple memoryless strategy that was easier to verify, compared to the strategy considered in the verification task.

Finally, regarding the invariant size parameter, we used $N_I = 1$ because it was sufficient for all our benchmarks. We also ran our prototype tool with $N_I = 2$ on $12$ of the benchmarks and $8$ of them were solved within the timeout of 5 minutes. The timeouts are likely due to the larger size of the final system of constraints. Indeed, given more time, we expect our method can be effectively applied with larger template sizes as well.

\begin{table}[t]
    \centering
    \resizebox{\textwidth}{!}{
    \begin{tabular}{|c|c|c|c|c|c|c|c|c|}
        \hline
        Model & Specification & Task & Coeff \# & Const \# & 	PQ Coeff \# & Query \# & SMT time & Total time \\ 
        \hline
        GW (3*3) & G F "V5>=0.9" & Synth & $62$ & $41$ & $171$ & $34$ & $<2s$ & $6s$ \\ 
        GW (3*3) & G F "V5>=0.9" \& G "V4<=0.5" & Synth & $80$ & $45$  & $277$ & $49$ & $<5s$ & $<5s$ \\ 
        GW (4*4) & G F "V11>=0.9" & Synth & $121$ & $79$ & $286$ & $81$ &$<10s$ & $10s$ \\ 
        GW (4*4)   & G F "V11>=0.9" \& G "V9<=0.5" & Synth & $153$ & $83$ & $448$ & $87$ &  $12s$ & $13s$ \\ 
        GW (5*5)   & G F "V19>=0.9" & Synth & $198$ & $129$ & $435$ & $131$ & $302s$ & $303s$ \\ 
        \hline
        PageRank & F G "V2>0.2" & Verify  & $60$ & $13$ & $400$ & $35$ & $63s$ & $64s$ \\ 
        PageRank & G ( "0.2<=V0" → F "0.2<=V2") & Verify & $48$ & $13$ & $253$ & $22$ & $17s$ & $18s$
        \\
        PageRank & G ( "0.2<=V2" → F "0.2<=V2") & Verify & $48$ & $13$ & $253$ & $22$ & $8s$ & $8s$
        \\
        PageRank & G ( "0.2<=V3" → F "0.2<=V2") & Verify & $48$ & $13$ & $253$ & $22$ &  $44s$ & $45s$
        \\
        PageRank & G ( "0.2<=V4" → F "0.2<=V2") & Verify &
        $48$ & $13$ & $253$ & $22$ &  $5s$ & $6s$
        \\
        PageRank & G "V0>=0.2" | "V1>=0.2"  | "V2>=0.2"  & Verify & $60$ & $19$ & $569$ & $44$ & $136s$ & $137s$
        \\
        &  | "V3>=0.2" | "V4>=0.2" → F  "V2=1" & & & & & & & 
        \\
        PageRank & F "V2=1"→ G "V1<=0.2" & Verify & $144$ & $35$ & $630$ & $46$ &  $6s$ & $6s$
        \\
        \hline
        
        Pharmacokinetics & F "V4=1" & Verify  & $42$ & $9$ & $176$ & $14$ & $<1s$ & $<1s$ \\   
        Pharmacokinetics & G ("0.13<=V3<=0.2" | "0<=V3<0.13") & Verify & $56$ & $11$ & $365$ & $28$ & $102s$ & $102s$
        \\
        \hline
        CAV23~\cite{AkshayCMZ23} & G F "V1>=0.249" & Verify & $16$ & $7$ & $85$ & $13$ &  $<2s$ & $<2s$   \\ 
        CAV23~\cite{AkshayCMZ23} & G F "V1>=0.249" & Synth & $20$ & $14$  & $108$ & $16$ & $7s$ & $<2s$
        \\
        CAV23~\cite{AkshayCMZ23} & "V1>= 0.249" U "V2 >= 0.25" & Verify & $32$ & $13$ & $185$ & $22$ & $<5s$ & $7s$
        \\
        CAV23~\cite{AkshayCMZ23} & "V1>= 0.249" U "V2 >= 0.25" & Synth & $36$ & $20$ & $189$ & $25$ & $7s$ & $5s$
        \\
        CAV23~\cite{AkshayCMZ23} & "0.334>=V1>=0.332" U "V0=0.25" & Verify & $32$ & $17$ & $229$ & $24$ & $<4s$ & $4s$
        \\
        CAV23~\cite{AkshayCMZ23} &  "0.334>=V1>=0.332" U "V0=0.25" & Synth & $36$ & $20$ & $233$ & $28$ & $<1s$ & $<1s$  \\         
        \hline  
    \end{tabular}
    }
    \caption{For each experiment we report, from left to right, the 
    benchmark, specification, task (verification or synthesis), the number of coefficients, the number of constraints, the number of coefficients in PolyQEnt generated file (i.e.~after application of Farkas' lemma), the number of constraints in PolyQEnt generated file, SMT-solving time, and the total runtime.}
    \label{tab:example}
\end{table}

\section{Conclusion}\label{sec:conclusion}

In this paper, we considered distributional $\omega$-regular specifications in MDPs and addressed the problems of strategy verification and synthesis. We developed new notions of product distributional transition systems between an MDP and an NBA. We then introduced distributional certificates, using which we provided template-based synthesis algorithms for strategy verification and synthesis. Our experiments demonstrate the benefits and promise of our approach. As future work, we would like to go beyond MDPs and consider partially observable MDPs. Moreover, it would be interesting to lift the objectives from NBA to Rabin automata, where even the notion of distributional certificates is unclear. 

\bibliography{bib}

\begin{thebibliography}{10}

\bibitem{DBLP:conf/cav/AbateGR24}
Alessandro Abate, Mirco Giacobbe, and Diptarko Roy.
\newblock Stochastic omega-regular verification and control with
  supermartingales.
\newblock In Arie Gurfinkel and Vijay Ganesh, editors, {\em Computer Aided
  Verification - 36th International Conference, {CAV} 2024, Montreal, QC,
  Canada, July 24-27, 2024, Proceedings, Part {III}}, volume 14683 of {\em
  Lecture Notes in Computer Science}, pages 395--419. Springer, 2024.
\newblock \href {https://doi.org/10.1007/978-3-031-65633-0\_18}
  {\path{doi:10.1007/978-3-031-65633-0\_18}}.

\bibitem{DBLP:journals/jacm/AgrawalAGT15}
Manindra Agrawal, S.~Akshay, Blaise Genest, and P.~S. Thiagarajan.
\newblock Approximate verification of the symbolic dynamics of {M}arkov chains.
\newblock {\em J. {ACM}}, 62(1):2:1--2:34, 2015.
\newblock \href {https://doi.org/10.1145/2629417} {\path{doi:10.1145/2629417}}.

\bibitem{AkshayAOW15}
S.~Akshay, Timos Antonopoulos, Jo{\"{e}}l Ouaknine, and James Worrell.
\newblock Reachability problems for {M}arkov chains.
\newblock {\em Inf. Process. Lett.}, 115(2):155--158, 2015.
\newblock URL: \url{https://doi.org/10.1016/j.ipl.2014.08.013}, \href
  {https://doi.org/10.1016/J.IPL.2014.08.013}
  {\path{doi:10.1016/J.IPL.2014.08.013}}.

\bibitem{AkshayCMZ23}
S.~Akshay, Krishnendu Chatterjee, Tobias Meggendorfer, and Dorde Zikelic.
\newblock {MDP}s as distribution transformers: Affine invariant synthesis for
  safety objectives.
\newblock In Constantin Enea and Akash Lal, editors, {\em Computer Aided
  Verification - 35th International Conference, {CAV} 2023, Paris, France, July
  17-22, 2023, Proceedings, Part {III}}, volume 13966 of {\em Lecture Notes in
  Computer Science}, pages 86--112. Springer, 2023.
\newblock \href {https://doi.org/10.1007/978-3-031-37709-9\_5}
  {\path{doi:10.1007/978-3-031-37709-9\_5}}.

\bibitem{AkshayCMZ24}
S.~Akshay, Krishnendu Chatterjee, Tobias Meggendorfer, and Dorde Zikelic.
\newblock Certified policy verification and synthesis for {MDP}s under
  distributional reach-avoidance properties.
\newblock In {\em Proceedings of the Thirty-Third International Joint
  Conference on Artificial Intelligence, {IJCAI} 2024, Jeju, South Korea,
  August 3-9, 2024}, pages 3--12. ijcai.org, 2024.
\newblock URL: \url{https://www.ijcai.org/proceedings/2024/1}.

\bibitem{DBLP:conf/lics/AkshayGV18}
S.~Akshay, Blaise Genest, and Nikhil Vyas.
\newblock Distribution-based objectives for {M}arkov decision processes.
\newblock In Anuj Dawar and Erich Gr{\"{a}}del, editors, {\em Proceedings of
  the 33rd Annual {ACM/IEEE} Symposium on Logic in Computer Science, {LICS}
  2018, Oxford, UK, July 09-12, 2018}, pages 36--45. {ACM}, 2018.
\newblock \href {https://doi.org/10.1145/3209108.3209185}
  {\path{doi:10.1145/3209108.3209185}}.

\bibitem{AsadiC0GM21}
Ali Asadi, Krishnendu Chatterjee, Hongfei Fu, Amir~Kafshdar Goharshady, and
  Mohammad Mahdavi.
\newblock Polynomial reachability witnesses via stellens{\"{a}}tze.
\newblock In Stephen~N. Freund and Eran Yahav, editors, {\em {PLDI} '21: 42nd
  {ACM} {SIGPLAN} International Conference on Programming Language Design and
  Implementation, Virtual Event, Canada, June 20-25, 2021}, pages 772--787.
  {ACM}, 2021.
\newblock \href {https://doi.org/10.1145/3453483.3454076}
  {\path{doi:10.1145/3453483.3454076}}.

\bibitem{BaierK08}
Christel Baier and Joost{-}Pieter Katoen.
\newblock {\em Principles of model checking}.
\newblock {MIT} Press, 2008.

\bibitem{BaldoniBMR08}
Roberto Baldoni, Fran{\c{c}}ois Bonnet, Alessia Milani, and Michel Raynal.
\newblock On the solvability of anonymous partial grids exploration by mobile
  robots.
\newblock In Theodore~P. Baker, Alain Bui, and S{\'{e}}bastien Tixeuil,
  editors, {\em Principles of Distributed Systems, 12th International
  Conference, {OPODIS} 2008, Luxor, Egypt, December 15-18, 2008. Proceedings},
  volume 5401 of {\em Lecture Notes in Computer Science}, pages 428--445.
  Springer, 2008.
\newblock \href {https://doi.org/10.1007/978-3-540-92221-6\_27}
  {\path{doi:10.1007/978-3-540-92221-6\_27}}.

\bibitem{DBLP:journals/logcom/BeauquierRS06}
Dani{\`{e}}le Beauquier, Alexander~Moshe Rabinovich, and Anatol Slissenko.
\newblock A logic of probability with decidable model checking.
\newblock {\em J. Log. Comput.}, 16(4):461--487, 2006.
\newblock \href {https://doi.org/10.1093/logcom/exl004}
  {\path{doi:10.1093/logcom/exl004}}.

\bibitem{DBLP:conf/qest/ChadhaKVAK11}
Rohit Chadha, Vijay~Anand Korthikanti, Mahesh Viswanathan, Gul Agha, and
  YoungMin Kwon.
\newblock Model checking {MDP}s with a unique compact invariant set of
  distributions.
\newblock In {\em Eighth International Conference on Quantitative Evaluation of
  Systems, {QEST} 2011, Aachen, Germany, 5-8 September, 2011}, pages 121--130.
  {IEEE} Computer Society, 2011.
\newblock \href {https://doi.org/10.1109/QEST.2011.22}
  {\path{doi:10.1109/QEST.2011.22}}.

\bibitem{SriramCAV}
Aleksandar Chakarov and Sriram Sankaranarayanan.
\newblock Probabilistic program analysis with martingales.
\newblock In Natasha Sharygina and Helmut Veith, editors, {\em Computer Aided
  Verification - 25th International Conference, {CAV} 2013, Saint Petersburg,
  Russia, July 13-19, 2013. Proceedings}, volume 8044 of {\em Lecture Notes in
  Computer Science}, pages 511--526. Springer, 2013.
\newblock \href {https://doi.org/10.1007/978-3-642-39799-8\_34}
  {\path{doi:10.1007/978-3-642-39799-8\_34}}.

\bibitem{ChatterjeeFG16}
Krishnendu Chatterjee, Hongfei Fu, and Amir~Kafshdar Goharshady.
\newblock Termination analysis of probabilistic programs through
  positivstellensatz's.
\newblock In {\em {CAV} {(1)}}, volume 9779 of {\em Lecture Notes in Computer
  Science}, pages 3--22. Springer, 2016.

\bibitem{DBLP:conf/pldi/Chatterjee0GG20}
Krishnendu Chatterjee, Hongfei Fu, Amir~Kafshdar Goharshady, and Ehsan~Kafshdar
  Goharshady.
\newblock Polynomial invariant generation for non-deterministic recursive
  programs.
\newblock In Alastair~F. Donaldson and Emina Torlak, editors, {\em Proceedings
  of the 41st {ACM} {SIGPLAN} International Conference on Programming Language
  Design and Implementation, {PLDI} 2020, London, UK, June 15-20, 2020}, pages
  672--687. {ACM}, 2020.
\newblock \href {https://doi.org/10.1145/3385412.3385969}
  {\path{doi:10.1145/3385412.3385969}}.

\bibitem{CFNH16:prob-termination}
Krishnendu Chatterjee, Hongfei Fu, Petr Novotn{\'{y}}, and Rouzbeh
  Hasheminezhad.
\newblock Algorithmic analysis of qualitative and quantitative termination
  problems for affine probabilistic programs.
\newblock {\em TOPLAS}, 40(2):7:1--7:45, 2018.
\newblock \href {https://doi.org/10.1145/3174800} {\path{doi:10.1145/3174800}}.

\bibitem{chatterjee2025polyqent}
Krishnendu Chatterjee, Amir~Kafshdar Goharshady, Ehsan~Kafshdar Goharshady,
  Mehrdad Karrabi, Milad Saadat, Maximilian Seeliger, and Đorđe Žikelić.
\newblock Polyqent: A polynomial quantified entailment solver, 2025.
\newblock URL: \url{https://arxiv.org/abs/2408.03796}, \href
  {http://arxiv.org/abs/2408.03796} {\path{arXiv:2408.03796}}.

\bibitem{ChatterjeeGGKZ24}
Krishnendu Chatterjee, Amir~Kafshdar Goharshady, Ehsan~Kafshdar Goharshady,
  Mehrdad Karrabi, and Dorde Zikelic.
\newblock Sound and complete witnesses for template-based verification of {LTL}
  properties on polynomial programs.
\newblock In Andr{\'{e}} Platzer, Kristin~Yvonne Rozier, Matteo Pradella, and
  Matteo Rossi, editors, {\em Formal Methods - 26th International Symposium,
  {FM} 2024, Milan, Italy, September 9-13, 2024, Proceedings, Part {I}}, volume
  14933 of {\em Lecture Notes in Computer Science}, pages 600--619. Springer,
  2024.
\newblock \href {https://doi.org/10.1007/978-3-031-71162-6\_31}
  {\path{doi:10.1007/978-3-031-71162-6\_31}}.

\bibitem{ChatterjeeGMZ22}
Krishnendu Chatterjee, Amir~Kafshdar Goharshady, Tobias Meggendorfer, and Dorde
  Zikelic.
\newblock Sound and complete certificates for quantitative termination analysis
  of probabilistic programs.
\newblock In {\em {CAV} {(1)}}, volume 13371 of {\em Lecture Notes in Computer
  Science}, pages 55--78. Springer, 2022.

\bibitem{CNZ17}
Krishnendu Chatterjee, Petr Novotn\'{y}, and {\DJ}or{\dj}e \v{Z}ikeli\'c.
\newblock Stochastic invariants for probabilistic termination.
\newblock In {\em POPL}, pages 145--160, 2017.
\newblock \href {https://doi.org/10.1145/3009837.3009873}
  {\path{doi:10.1145/3009837.3009873}}.

\bibitem{mathsat5}
Alessandro Cimatti, Alberto Griggio, Bastiaan Schaafsma, and Roberto
  Sebastiani.
\newblock {The MathSAT5 SMT Solver}.
\newblock In Nir Piterman and Scott Smolka, editors, {\em Proceedings of
  TACAS}, volume 7795 of {\em LNCS}. Springer, 2013.

\bibitem{ColonSS03}
Michael Col{\'{o}}n, Sriram Sankaranarayanan, and Henny Sipma.
\newblock Linear invariant generation using non-linear constraint solving.
\newblock In Warren A.~Hunt Jr. and Fabio Somenzi, editors, {\em Computer Aided
  Verification, 15th International Conference, {CAV} 2003, Boulder, CO, USA,
  July 8-12, 2003, Proceedings}, volume 2725 of {\em Lecture Notes in Computer
  Science}, pages 420--432. Springer, 2003.
\newblock \href {https://doi.org/10.1007/978-3-540-45069-6\_39}
  {\path{doi:10.1007/978-3-540-45069-6\_39}}.

\bibitem{ColonS01}
Michael Col{\'{o}}n and Henny Sipma.
\newblock Synthesis of linear ranking functions.
\newblock In Tiziana Margaria and Wang Yi, editors, {\em Tools and Algorithms
  for the Construction and Analysis of Systems, 7th International Conference,
  {TACAS} 2001 Held as Part of the Joint European Conferences on Theory and
  Practice of Software, {ETAPS} 2001 Genova, Italy, April 2-6, 2001,
  Proceedings}, volume 2031 of {\em Lecture Notes in Computer Science}, pages
  67--81. Springer, 2001.
\newblock \href {https://doi.org/10.1007/3-540-45319-9\_6}
  {\path{doi:10.1007/3-540-45319-9\_6}}.

\bibitem{MouraB08}
Leonardo~Mendon{\c{c}}a de~Moura and Nikolaj~S. Bj{\o}rner.
\newblock {Z3:} an efficient {SMT} solver.
\newblock In {\em {TACAS}}, volume 4963 of {\em Lecture Notes in Computer
  Science}, pages 337--340. Springer, 2008.

\bibitem{duret2022spot}
Alexandre Duret-Lutz, Etienne Renault, Maximilien Colange, Florian Renkin,
  Alexandre Gbaguidi~Aisse, Philipp Schlehuber-Caissier, Thomas Medioni,
  Antoine Martin, J{\'e}r{\^o}me Dubois, Cl{\'e}ment Gillard, et~al.
\newblock From spot 2.0 to spot 2.10: What’s new?
\newblock In {\em International Conference on Computer Aided Verification},
  pages 174--187. Springer, 2022.

\bibitem{farkas1902theorie}
Julius Farkas.
\newblock Theorie der einfachen ungleichungen.
\newblock {\em Journal f{\"u}r die reine und angewandte Mathematik (Crelles
  Journal)}, 1902(124):1--27, 1902.

\bibitem{DBLP:journals/tac/GaoAXJ24}
Yulong Gao, Alessandro Abate, Lihua Xie, and Karl~Henrik Johansson.
\newblock Distributional reachability for {M}arkov decision processes: Theory
  and applications.
\newblock {\em {IEEE} Trans. Autom. Control.}, 69(7):4598--4613, 2024.
\newblock URL: \url{https://doi.org/10.1109/tac.2023.3341282}, \href
  {https://doi.org/10.1109/TAC.2023.3341282}
  {\path{doi:10.1109/TAC.2023.3341282}}.

\bibitem{handelman1988representing}
David Handelman.
\newblock Representing polynomials by positive linear functions on compact
  convex polyhedra.
\newblock {\em Pacific Journal of Mathematics}, 132(1):35--62, 1988.

\bibitem{HenzingerMW09}
Thomas~A. Henzinger, Maria Mateescu, and Verena Wolf.
\newblock Sliding window abstraction for infinite {M}arkov chains.
\newblock In Ahmed Bouajjani and Oded Maler, editors, {\em Computer Aided
  Verification, 21st International Conference, {CAV} 2009, Grenoble, France,
  June 26 - July 2, 2009. Proceedings}, volume 5643 of {\em Lecture Notes in
  Computer Science}, pages 337--352. Springer, 2009.
\newblock \href {https://doi.org/10.1007/978-3-642-02658-4\_27}
  {\path{doi:10.1007/978-3-642-02658-4\_27}}.

\bibitem{DBLP:conf/qest/KorthikantiVAK10}
Vijay~Anand Korthikanti, Mahesh Viswanathan, Gul Agha, and YoungMin Kwon.
\newblock Reasoning about {MDP}s as transformers of probability distributions.
\newblock In {\em {QEST} 2010, Seventh International Conference on the
  Quantitative Evaluation of Systems, Williamsburg, Virginia, USA, 15-18
  September 2010}, pages 199--208. {IEEE} Computer Society, 2010.
\newblock \href {https://doi.org/10.1109/QEST.2010.35}
  {\path{doi:10.1109/QEST.2010.35}}.

\bibitem{KwiatkowskaNP11}
Marta~Z. Kwiatkowska, Gethin Norman, and David Parker.
\newblock {PRISM} 4.0: Verification of probabilistic real-time systems.
\newblock In {\em {CAV}}, volume 6806 of {\em Lecture Notes in Computer
  Science}, pages 585--591. Springer, 2011.

\bibitem{DBLP:journals/tse/KwonA11}
YoungMin Kwon and Gul~A. Agha.
\newblock Verifying the evolution of probability distributions governed by a
  {DTMC}.
\newblock {\em {IEEE} Trans. Software Eng.}, 37(1):126--141, 2011.
\newblock \href {https://doi.org/10.1109/TSE.2010.80}
  {\path{doi:10.1109/TSE.2010.80}}.

\bibitem{DBLP:conf/rp/OuaknineW12}
Jo{\"{e}}l Ouaknine and James Worrell.
\newblock Decision problems for linear recurrence sequences.
\newblock In Alain Finkel, J{\'{e}}r{\^{o}}me Leroux, and Igor Potapov,
  editors, {\em Reachability Problems - 6th International Workshop, {RP} 2012,
  Bordeaux, France, September 17-19, 2012. Proceedings}, volume 7550 of {\em
  Lecture Notes in Computer Science}, pages 21--28. Springer, 2012.
\newblock \href {https://doi.org/10.1007/978-3-642-33512-9\_3}
  {\path{doi:10.1007/978-3-642-33512-9\_3}}.

\bibitem{PrajnaJP07}
Stephen Prajna, Ali Jadbabaie, and George~J. Pappas.
\newblock A framework for worst-case and stochastic safety verification using
  barrier certificates.
\newblock {\em {IEEE} Trans. Autom. Control.}, 52(8):1415--1428, 2007.
\newblock \href {https://doi.org/10.1109/TAC.2007.902736}
  {\path{doi:10.1109/TAC.2007.902736}}.

\bibitem{takisaka2021ranking}
Toru Takisaka, Yuichiro Oyabu, Natsuki Urabe, and Ichiro Hasuo.
\newblock Ranking and repulsing supermartingales for reachability in randomized
  programs.
\newblock {\em {ACM} Trans. Program. Lang. Syst.}, 43(2):5:1--5:46, 2021.
\newblock \href {https://doi.org/10.1145/3450967} {\path{doi:10.1145/3450967}}.

\bibitem{ZikelicLHC23}
Dorde Zikelic, Mathias Lechner, Thomas~A. Henzinger, and Krishnendu Chatterjee.
\newblock Learning control policies for stochastic systems with reach-avoid
  guarantees.
\newblock In Brian Williams, Yiling Chen, and Jennifer Neville, editors, {\em
  Thirty-Seventh {AAAI} Conference on Artificial Intelligence, {AAAI} 2023,
  Thirty-Fifth Conference on Innovative Applications of Artificial
  Intelligence, {IAAI} 2023, Thirteenth Symposium on Educational Advances in
  Artificial Intelligence, {EAAI} 2023, Washington, DC, USA, February 7-14,
  2023}, pages 11926--11935. {AAAI} Press, 2023.
\newblock URL: \url{https://ojs.aaai.org/index.php/AAAI/article/view/26407}.

\end{thebibliography}

\end{document}